\documentclass[sigconf,nonacm]{acmart}

\usepackage{amsmath}
\usepackage{mathtools}
\usepackage{amsfonts}

\usepackage{amssymb}
\usepackage{natbib}
\usepackage{graphicx}
\usepackage{booktabs}

\usepackage{tikz}
\usetikzlibrary{arrows.meta,positioning}

\graphicspath{{Figs/}}

\newtheorem{assumption}{Assumption}
\newtheorem{proposition}{Proposition}
\newtheorem{remark}{Remark}
\newtheorem{definition}{Definition}

\setcopyright{none}
\renewcommand\footnotetextcopyrightpermission[1]{}
\begin{document}

\title{Statistical Foundations of LLM-based A/B Testing: \\ A Surrogacy Framework for Human Causal Inference}

\author{Joel Persson}
\authornote{Corresponding author.}
\affiliation{%
  \institution{Spotify USA, Inc.}
  \city{New York}
  \state{NY}
  \country{USA}
}
\email{joelpersson@spotify.com}

\author{Mårten Schultzberg}
\affiliation{%
  \institution{Spotify AB}
  \city{Stockholm}
  \country{Sweden}
}
\email{mschultzberg@spotify.com}

\author{Sebastian Ankargren}
\affiliation{%
  \institution{Spotify AB}
  \city{Stockholm}
  \country{Sweden}
}
\email{sebastiana@spotify.com}

\begin{abstract}
Organizations and researchers show increasing interest in using large language models (LLMs) in place of human participants in A/B tests, in the hope of experimenting faster and at lower cost. We study when a treatment effect estimated on LLM outcomes can recover the effect for the human population of interest. Distributional equivalence between LLM and human outcomes would make any standard estimator valid but is unrealistic. We therefore develop a statistical framework that adapts surrogate endpoint theory to LLMs, showing that calibrating LLM outcomes to human outcomes identifies the average treatment effect under surrogacy and comparability conditions that are jointly weaker than distributional equivalence. We present a falsification test for surrogacy and a bound on the worst-case bias from limited overlap between the LLM and human samples. We further show that the stochasticity inherent to LLMs can weaken surrogacy for identification while also introducing bias and variance during estimation, but that using an average over multiple LLM draws per unit as the surrogate mitigates these issues. Simulations validate the results, and an empirical application to the Upworthy Research Archive dataset shows that raw LLM outputs recover only 39\% of the human treatment effect while nonparametric calibration closes the gap. A central takeaway is that A/B testing on LLM responses is correct only by assumption, whereas A/B testing on humans is correct by design, and that the required assumptions are hardest to justify precisely where LLMs promise the greatest benefit. We discuss the choice of LLM, prompting, and temperature as design variables, the compounded challenge posed by long-term outcomes, and how to size human pilot studies for validation.
\end{abstract}


\maketitle
\thispagestyle{plain}

\section{Introduction}
Large language models (LLMs) increasingly serve as proxies for human participants in research, offering the promise of rapid hypothesis testing and cost-effective experimentation. This shift toward AI-mediated research presents both unprecedented opportunities and fundamental challenges for causal inference. On the one hand, LLMs appear to enable the simulation of arbitrarily large datasets and even unit-level counterfactual outcomes, seemingly sidestepping classical challenges in causal inference. However, such simulations will tend to recover treatment effects for the LLM-generated outcomes themselves, which need not coincide with treatment effects for the human populations of interest. This distinction gives rise to a central question: under what conditions can outcomes generated by LLMs be used to identify treatment effects for the human populations of interest?

In this paper, we study this question from the lens of surrogate-endpoint theory \citep{prentice1989surrogate}, by viewing LLM-generated outcomes as noisy surrogates for human outcomes and asking when a treatment effect measured on LLM samples recovers the human treatment effect of interest. Our goal is to make explicit what conditions must hold for experimentation with LLM outcomes to yield the intended results, resulting in clear and principled recommendations for practice. In what follows, we first situate our work in the literature on LLMs as stand-ins for humans and then the theory it builds on.

Recent work demonstrates that LLM-based agents can partially replicate human behavior with notable accuracy in specific settings, suggesting that they may serve as useful proxies in experiments. Theoretically grounded ``general'' agents, i.e., LLM personas guided by social-science theory and optimized on related data, predict human behavior in previously unseen strategic environments, often outperforming simple game-theoretic benchmarks and, in some cases, even surpassing prior human datasets for prediction \citep{manning2025general}. GPT-4 style models forecast outcomes of social-science experiments with high accuracy, with model-based forecasts strongly correlating with realized effects across many studies, including some conducted after the model's training cutoff \citep{hewitt2024predicting}. Simulated multi-human samples replicate a range of classic findings from economics and psychology, though not without systematic distortions that deviate from human patterns \citep{aher2023using}. Adjacent domains show similar promise. Persona-conditioned LLM panels emulate professional forecasters with distributional properties close to human panels and, in some cases, improved accuracy \citep{hansen2024spf}. Large-scale fine-tuning on survey responses improves alignment with real public-opinion distributions and generalizes to new surveys and demographics \citep{suh2025publicopinion}.

However, important limitations emerge when moving from prediction to causal inference and from narrow tasks to broad generalization. Even in simple behavioral games, state-of-the-art prompting or fine-tuning fails to reproduce the empirical distribution of human choices, yielding overly peaky or stereotyped response profiles that diverge from human data \citep{gao2025scylla}. From a design perspective, using LLM ``participants'' in experiments introduces novel types of confounding: concealment of treatment/control to preserve human-like blindness may cause the model to miss the true treatment channel, while revealing design details can make the model over-index on artifacts of the setup, undermining realism \citep{gui2023challenge}. Recent methodological work in social sciences develops frameworks for combining imperfect LLM annotations with smaller samples of high-quality human labels to maintain valid statistical inference, an approach that shows promise but requires careful calibration and validation \citep{egami2023imperfect}.

The balance of evidence suggests a nuanced role for LLM agents in experimental pipelines. While theory-guided prompting and principled optimization can improve LLMs' predictive accuracy and generalization \citep{manning2025general}, and domain-specific tailoring can bring simulated distributions closer to human ground truth \citep{hansen2024spf,suh2025publicopinion}, well-documented sources of bias (such as stereotyped variability, prompt sensitivity, and causal inference challenges) underscore that current LLMs remain, at best, partial proxies whose validity depends on the domain, outcome, and design \citep{gao2025scylla,gui2023challenge,egami2023imperfect}.

AgentA/B \citep{agentab2024} demonstrates a scalable agent-based A/B test on a live website (Amazon.com), assigning 1,000 LLM agents to treatment and control variants. The system produces directional alignment with human behavior, in that LLM agents detect subtle design differences (such as changes to filter panel layouts) and generate measurably different outcomes across treatment and control groups. However, it also exhibits systematic differences to humans in terms of exploration and efficiency. This illustrates both the promise of LLM-based experimentation and the central challenge: while LLMs can dramatically reduce both experimental costs and timeline requirements compared to traditional A/B testing, it remains unclear under what statistical conditions LLM responses can be expected to reliably identify treatment effects for the human populations of interest. A surrogacy framework, which treats the LLM outcome as an imperfect stand-in for the human outcome, gives a way to make these conditions precise.

We formalize this surrogacy view and derive the conditions under which treatment effects for humans can be reliably identified and estimated from LLM-based experiments. In Section~\ref{sec:setup}, we introduce the experimental setup and the LLM surrogate as a stochastic mapping. We clarify that perfect replication of human behavior (distributional equivalence between LLM and human outcomes) is sufficient but not necessary for identification. Section~\ref{sec:framework} then develops the surrogacy framework, showing that point-identification can be achieved through a weaker calibration relationship between LLM-generated and human outcomes, provided that certain surrogacy and comparability assumptions hold. Section~\ref{sec:estimation_stochastic} then studies estimation, including the consequences of stochasticity in LLM outputs. We show that using just a single draw per unit can undermine surrogacy and will lead to attenuation bias and variance inflation, but that averaging over multiple draws mitigates all three problems. Our approach to identification and estimation is inspired by the surrogate index approach of \citet{athey2019using} and related methods \citep{athey2025usingexperimentscorrectselection, kallus2020role}, which we adapt to the context of LLM-generated outcomes, where the surrogate is something the experimenter generates and controls rather than passively observes. We also draw on classical errors-in-variables results \citep{fuller1987measurement} to arrive at the consequences of stochasticity in LLMs. In Section~\ref{sec:evaluation}, we propose how to check the quality of the identification and estimation, resulting in a moment-condition test of the surrogacy relationship and a theoretical and estimable bound that quantifies the worst-case bias in the treatment effect as a consequence of violations to distributional overlap between LLM and human samples. Section~\ref{sec:simulation} validates our theoretical results via simulation studies, whereas in Section~\ref{sec:empirical}, we apply the framework to real data from the Upworthy Research Archive, finding that raw LLM predictions recover only 39\% of the human treatment effect while nonparametric calibration closes the gap. In Section~\ref{sec:design}, we provide practical guidance on the design and deployment of LLM experiments, including the role of prompting, temperature, long-term outcomes, and pilot studies. Section~\ref{sec:discussion} concludes.

A main takeaway of our work is to highlight fundamental limitations: while necessary assumptions can be partially assessed using historical data, identification of causal effects for new interventions cannot, in general, be empirically verified without human outcomes. We interpret this to mean that human verification remains indispensable for novel inferences.

\section{Setup}\label{sec:setup}

We consider randomized experiments with two arms. Let $W\in\{0, 1\}$ denote treatment assignment with $W=0$ indicating assignment into the control group, $X$ pre-treatment covariates, and $Y$ the outcome. We use the potential outcomes framework \citep{Neyman1923_EnglishTrans1990, Rubin1974} and let $Y(W)$ be the outcome under treatment status $W$. The estimand of interest is the average treatment effect (ATE), defined as
\begin{equation}
    \tau \coloneqq \mathbb{E}[Y(1) - Y(0)].
\end{equation}
Given observations of the tuple $(W, X, Y)$ and randomized treatment assignment, the ATE is identified. The aim of this paper is not to identify $\tau$ from human outcomes, but to do so when the human outcome $Y$ is augmented or replaced by an outcome generated by a large language model (LLM) in response to the same inputs $(W, X)$. This requires notation for how an LLM produces such an outcome, which we introduce next.

An LLM generates its output by stochastically sampling a token sequence under a decoding configuration $\mathcal{D}$ such as the sampling temperature \citep{holtzman2020curious}, and the outcome $Y^*$ is obtained by post-processing this sequence into a numerical or behavioral response. Because the sampling is stochastic, $Y^*$ is a random variable whose distribution is determined by the model $M$, the prompt $I$, the decoding configuration $\mathcal{D}$, and the conditioning inputs $(W, X)$.

\begin{definition}[LLM surrogate]
For a fixed model $M$, prompt template $I$, and decoding configuration $\mathcal{D}$, the LLM surrogate is a random variable
\begin{equation}\label{eq:Ystar}
    Y^* \sim F_{M,I,\mathcal{D}}(\cdot \mid W, X).
\end{equation}
Throughout, $(M, I, \mathcal{D})$ are assumed fixed and omitted from the notation. We write $Y^* \mid (W, X) \sim F(\cdot \mid W, X)$.
\end{definition}

\begin{remark}[Deterministic special case]\label{rem:deterministic}
When $F(\cdot \mid W, X)$ is a point mass at some value $f(W, X)$, as arises under greedy decoding or zero temperature, the surrogate reduces to the deterministic specification $Y^* = f(W, X)$. All results in this paper also cover this case. The deterministic setting is thus nested within our general framework rather than treated as a separate modeling choice.
\end{remark}

Given access to an LLM, an LLM-based A/B test uses the following procedure:
\begin{enumerate}
    \item Instruct the LLM to adopt characteristics given by $X$.
    \item Let the LLM interact with the environment/scenario corresponding to its assigned condition ($W=0$ or $W=1$).
    \item Record the outcome $Y^*$, which is a draw from $F(\cdot \mid W, X)$.
\end{enumerate}

By repeating this procedure as many times as desired, the experimenter obtains an \emph{artificial sample} of tuples $(W, X, Y^*)$. Throughout, we assume access to two samples, drawn from populations indexed by an indicator $P$:
\begin{itemize}
\item the \emph{experimental sample} ($P=0$), consisting of human data $(W, X, Y^*, Y)$ in which both the surrogate and the human outcome are observed; and
\item the \emph{artificial sample} ($P=1$),\allowbreak{} consisting of LLM data $(W, X, Y^*)$ in which only the surrogate is observed.
\end{itemize}
In both samples treatment is independently randomized and SUTVA holds, by construction in the artificial sample and by design in the experiment. We hold $(M, I, \mathcal{D})$ fixed across the two samples, so that the conditional distribution $Y^* \mid (W, X) \sim F(\cdot \mid W, X)$ is identical under $P=0$ and $P=1$, even though the realized values of $Y^*$ differ across samples, and across draws within a sample, because the LLM is stochastic. We further assume that the marginal distribution of $X$ is the same under $P=1$ and $P=0$, so that the artificial sample represents the same covariate population as the human experiment. This is trivially satisfied when the LLM is prompted with user profiles drawn from the same population as the human experiment, as is the standard practice in LLM-based A/B testing.

The experimenter may also generate $K \geq 1$ independent draws from the LLM for each unit, yielding a collection of synthetic responses $Y^*_1, \ldots, Y^*_K$ where each $Y^*_k \mid (W, X) \overset{\text{iid}}{\sim} F(\cdot \mid W, X)$. We refer to $K$ as the \emph{replication count}. When $K = 1$ or the LLM is deterministic, the multiple draws reduce to a single surrogate outcome per unit.

Our central research question in this setting is: under what assumptions can we use the artificial sample $(W, X, Y^*)$ to infer the human ATE $\tau$ under $P=0$?

\section{A Surrogacy Framework for LLM-based A/B Testing}\label{sec:framework}

We now introduce our framework. We first cover the ideal but unrealistic case that the LLM responses perfectly replicate those of humans and then consider the realistic case that they do not, which leads us to our theoretical results and methods.

\subsection{The ideal case: perfect substitution}\label{sec:perfect}

The ideal scenario of LLM-based A/B testing is that the observed LLM outcomes are distributionally the same as human outcomes, 
\begin{equation}\label{eq:utopia}
    Y^* \stackrel{d}{=} Y(W) \mid X,
\end{equation}
or, minimally, that $\mathbb{E}[Y^* \mid X] = \mathbb{E}[Y(W)\mid X]$.

If this assumption holds, an LLM-generated sample is indistinguishable from a human sample, and arguments for identification and unbiased estimation of the ATE apply. However, this assumption is questionable in most cases. Even an LLM trained on user behavior may prioritize modeling the process (i.e., the sequence of actions) rather than ensuring that the final outcome $Y^*$ has the correct distribution or a functional thereof.

A natural and more principled approach to achieving this is to explicitly train or fine-tune the LLM so that the distribution of LLM responses matches the distribution of the human potential outcomes, or, if interest lies in the ATE, aligns in expectation. This can be viewed as minimizing a suitable divergence between the conditional distributions of $Y^*$ and $Y$ given $(W, X)$. Such approaches are closely related to recent work on aligning LLM outputs to human data via reinforcement learning from human feedback \citep{christiano2017deep, ouyang2022training}, or more general distribution-matching objectives \citep{yun2025alignment}.

Achieving such alignment in practice is still challenging and may not be possible to the desired degree of accuracy across different contexts where the LLM would substitute for human experiments (e.g., interventions on a homepage, in search, and across mobile and desktop users). As a result, the approach can be overly sensitive to the underlying data-generating process and the amount and quality of human data available for fine-tuning.

In the absence of perfect substitution between LLMs and human subjects, one can treat LLM outputs $Y^*$ as surrogate outcomes for the true human outcomes. Viewing LLM output from this lens makes it easier to reason about and build intuition for the requirements for replacing humans with LLMs.

\subsection{Identification: surrogacy and comparability}\label{sec:surrogacy}

We view $Y^*$ as a surrogate for $Y$ and ask when the ATE remains identified from $Y^*$ alone. Identification rests on two assumptions, surrogacy and comparability, together with a calibration function learned on human data. Figure~\ref{fig:assumptions} illustrates the two assumptions as causal diagrams. We then state each assumption formally, explain it, and give the identification result.

\begin{figure}[t]
\centering
\begin{tikzpicture}[
  >=Stealth, node distance=8mm and 12mm,
  nd/.style={circle, draw, inner sep=1pt, minimum size=8mm},
  arr/.style={->, thin},
  barr/.style={->, thin, dashed, gray}
]
  \node[nd] (W) {$W$};
  \node[nd, right=of W] (Ys) {$Y^*$};
  \node[nd, right=of Ys] (Y) {$Y$};
  \node[nd, below=of Ys] (X) {$X$};
  \draw[arr] (W) -- (Ys);
  \draw[arr] (Ys) -- (Y);
  \draw[arr] (X) -- (Ys);
  \draw[arr] (X) -- (Y);
  \draw[barr] (W) to[bend left=28] node[above, font=\scriptsize] {blocked} (Y);
  \node[anchor=south west, font=\small] at ([yshift=0.5mm]current bounding box.north west) {(a)};
\end{tikzpicture}

\vspace{2mm}

\begin{tikzpicture}[
  >=Stealth, node distance=8mm and 30mm,
  nd/.style={circle, draw, inner sep=1pt, minimum size=8mm},
  bx/.style={draw, rounded corners=2pt, inner sep=4pt},
  arr/.style={->, thin}
]
  \node[bx] (P0) {$P=0$};
  \node[bx, right=of P0] (P1) {$P=1$};
  \node[nd, below=of P0] (mu0) {$\mu$};
  \node[nd, below=of P1] (mu1) {$\mu$};
  \draw[arr] (P0) -- (mu0);
  \draw[arr] (P1) -- (mu1);
  \draw[<->, thin] (mu0) -- node[above, font=\scriptsize] {same} (mu1);
  \node[anchor=south west, font=\small] at ([yshift=0.5mm]current bounding box.north west) {(b)};
\end{tikzpicture}

\caption{Causal diagrams of the two identifying assumptions. Panel (a) shows surrogacy, under which $Y^*$ and $X$ fully mediate the effect of $W$ on $Y$, so that the direct path from $W$ to $Y$ (dashed) is blocked. Panel (b) shows comparability, under which the calibration function $\mu$ is the same in the experimental population $P=0$ and the artificial population $P=1$.}
\label{fig:assumptions}
\end{figure}
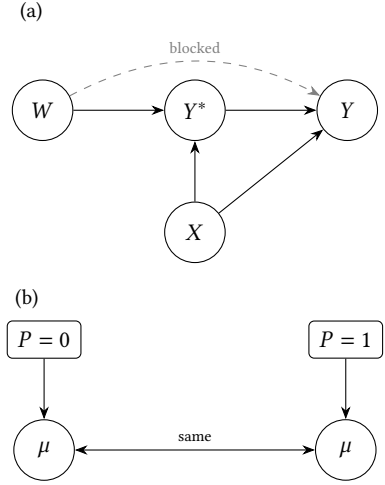

\begin{assumption}[Surrogacy, Prentice criterion]\label{as:surrogacy}
Conditional on the surrogate $Y^*$ and covariates $X$, treatment has no residual effect on $Y$:
$$Y \perp W \mid X, Y^*.$$
\end{assumption}

Surrogacy says that $Y^*$ fully mediates the effect of $W$ on $Y$. In other words, conditional on $(X, Y^*)$, the treatment status carries no further information about $Y$; see Figure~\ref{fig:assumptions}(a). The surrogate then stands in for $Y$, and any estimand of $Y$ can in principle be identified from $Y^*$. If only the ATE is of interest, the assumption can be weakened to the mean-independence condition $\mathbb E[Y \mid W, X, Y^*] = \mathbb E[Y \mid X, Y^*]$ \citep{athey2019using}. Whether surrogacy is plausible depends on how the LLM is built and trained, which we take up as a design choice in Section~\ref{sec:design}.

Surrogacy specifies what information $Y^*$ must carry, but not how to recover the human outcome from it. For the ATE, this amounts to learning the conditional mean of $Y$ given $(X, Y^*)$, i.e., the calibration function
\begin{align}\label{eq:calibration}
\mu(x, y^*) \coloneqq \mathbb{E}[Y \mid X=x, Y^*=y^*, P=0].
\end{align}
This function is to be estimated on the human sample $P=0$ from $(W, X, Y, Y^*)$ and evaluated on the artificial sample $P=1$. How well this identifies the ATE then depends on a second assumption.

\begin{assumption}[Comparability]\label{as:comparability}
The conditional distribution of $Y$ given $(X, Y^*)$ is stable across the human and artificial sample:
$$Y \perp P \mid X, Y^*,$$
with support overlap in $(X, Y^*)$.
\end{assumption}

Assumption~\ref{as:comparability} requires that $\mu(x, y^*)$ is invariant across $P$. Intuitively, it rules out any change in the relationship between $Y^*$ and $Y$ across the human and artificial samples. This invariance is shown in Figure~\ref{fig:assumptions}(b).

Under Assumptions~\ref{as:surrogacy}--\ref{as:comparability}, the ATE is identified from LLM data, following the analogous identification result in \citet{athey2019using}:

\begin{theorem}[Identification under surrogate transport]\label{thm:identification}
Under Assumptions~\ref{as:surrogacy}--\ref{as:comparability},
\begin{equation}\label{eq:identification}
    \tau
    = \mathbb E \left[ \mu(X, Y^*) \mid P=1, W=1 \right] - \mathbb E \left[ \mu(X, Y^*) \mid P=1, W=0 \right]
    .
    \end{equation}
\end{theorem}

\begin{proof}
By SUTVA and randomisation of $W$ in $P=0$, we have
$\mathbb{E}[Y(w) \mid X] = \mathbb{E}[Y \mid X, W\!=\!w, P\!=\!0]$.
Applying iterated expectations under Assumption~\ref{as:surrogacy},
\begin{align}
    \mathbb E[Y(w) \mid X]
    &= \mathbb E\!\left[\mathbb{E}[Y \mid X, Y^*]
      \,\big|\, X, W\!=\!w, P\!=\!0\right] \notag \\
    &= \mathbb E\left[\mu(X, Y^*) \mid X, W\!=\!w \right],
\end{align}
where the last equality uses Assumption~\ref{as:comparability} to drop
the conditioning on $P\!=\!0$, as $\mu(X,Y^*)$ is invariant across $P$, and the conditional distribution of $Y^* \mid (X,W)$ is the same under $P=0$ and $P=1$ (c.f.\ Section~\ref{sec:setup}). Marginalising over $X$ and using randomisation of $W$ also in
$P=1$ gives
$\mathbb{E}[\mu(X, Y^*) \mid W\!=\!w, P\!=\!1] = \mathbb E[Y(w)]$.
\end{proof}

Theorem~\ref{thm:identification} imposes a weaker condition on the LLM than perfect substitution (c.f. Eq.~\eqref{eq:utopia}). The calibration function $\mu$ absorbs any difference in scale or level between the surrogate and the human outcome, so the LLM need not reproduce human outcomes in absolute terms. The only requirement is that the treatment effect is mediated by $Y^*$ through surrogacy, and that the learned mapping $\mu$ transports from the experimental to the artificial population through comparability. As such, learning $\mu$ reduces to a standard regression (supervised learning) task of learning the mapping from $(X, Y^*)$ to $Y$ on $P=0$, while estimating the ATE reduces to a standard out-of-sample prediction task of evaluating that regression function on $P=1$. Best practices from the literature can be used to improve both steps, such as flexible machine-learning estimation and cross-fitting.

The natural sample analog of \eqref{eq:identification} is the plug-in estimator
\begin{equation}\label{eq:ate_estimator}
    \hat{\tau} = \frac{1}{n_1}\sum_{i:\,W_i=1,\,P_i=1}\hat\mu(X_i, Y^*_i) \;-\; \frac{1}{n_0}\sum_{i:\,W_i=0,\,P_i=1}\hat\mu(X_i, Y^*_i),
\end{equation}
where $n_w$ is the number of artificial-sample units with $W_i = w$. The cross-fitting of $\hat\mu$ recommended above ensures that its estimation error is asymptotically negligible \citep{chernozhukov2018double}.

\subsection{The role of stochasticity in LLM surrogates}\label{sec:estimation_stochastic}

What sets LLM surrogates apart from classical ones is that the experimenter may control their stochasticity as well as being able to draw a much larger number of them per unit. That is, the same inputs $(W, X)$ produce a distribution of outputs $F(\cdot \mid W, X)$ rather than a single value, and the experimenter chooses how many draws to take per unit. This section examines the consequences of the stochasticity and number of draws on first the surrogacy condition itself and then for the quality of the calibrated estimator.

Assumption~\ref{as:surrogacy} conditions on a single realized draw $Y^*$ from $F(\cdot \mid W, X)$. It therefore requires that one noisy LLM output, together with the covariates, carries all the treatment-effect information about $Y$. This is harder to satisfy when the sampling noise is large relative to the treatment-effect signal in $Y^*$. Sampling multiple draws per unit can help, though only when the treatment-effect signal is preserved in the latent mean of the surrogate, a condition we now make precise.

Consider the decomposition $Y^* = \theta(W, X) + \varepsilon$, where $\theta(W, X) \coloneqq \mathbb{E}[Y^* \mid W, X]$ is the latent conditional mean of the surrogate and $\varepsilon \coloneqq Y^* - \theta(W, X)$ is the per-draw noise, satisfying $\mathbb{E}[\varepsilon \mid W, X] = 0$ and $\mathrm{Var}(\varepsilon \mid W, X) \coloneqq \sigma^2_\varepsilon(W, X)$. The deterministic case of Remark~\ref{rem:deterministic} corresponds to $\sigma^2_\varepsilon = 0$.

With $K > 1$ draws per unit, the experimenter may condition on the averaged surrogate $\bar Y^*_K = K^{-1}\sum_{k=1}^K Y^*_k$ rather than a single draw. This can weaken the surrogacy requirement, as formalized in the following:

\begin{proposition}[Multi-draw surrogacy]\label{prop:multidraw}
Suppose $Y \perp W \mid X, \theta(W, X)$, $\varepsilon \perp Y \mid W, X$,
and that $\mu(x, \cdot)$ is continuous for almost every $x$. Then
$\bar Y^*_K \xrightarrow{a.s.} \theta(W, X)$ as $K \to \infty$, and
\[
  \mu(X, \bar Y^*_K)
  \;\xrightarrow{a.s.}\;
  \mu(X, \theta(W, X)),
\]
so that in the limit, $W$ affects $\mu(X, \cdot)$ only through $\theta(W, X)$.
\end{proposition}

\begin{proof}
The almost-sure convergence of $\bar Y^*_K$ to $\theta(W,X)$ follows
from the strong law of large numbers applied to the $K$ conditionally
i.i.d.\ draws of $Y^*$ given $(W,X)$. Since $\mu(x,\cdot)$ is
continuous by assumption, the continuous mapping theorem gives
$\mu(X, \bar Y^*_K) \xrightarrow{a.s.} \mu(X, \theta(W,X))$. As
$\mu(X, \theta(W,X))$ is a function of $(X, \theta(W,X))$, $W$ enters
the limit only through $\theta(W,X)$.
\end{proof}

Proposition~\ref{prop:multidraw} is a statement about identification: when surrogacy holds for the latent mean, averaging recovers a valid surrogate in the sampling limit. But the stochasticity also affects estimation. Even when surrogacy holds, noise in $Y^*$ degrades the calibrated estimator \eqref{eq:ate_estimator} in two ways, both mitigated by averaging. Taking $\mu$ as known for the moment, by the law of total variance the variance of each arm-specific mean decomposes into a signal component reflecting across-unit heterogeneity, which is present even when $Y^*$ is deterministic, and a noise component reflecting within-unit LLM stochasticity, which vanishes when $\sigma^2_\varepsilon = 0$. That is,
\begin{align}\label{eq:var_decomp}
    \mathrm{Var}\big(\mu(X, Y^*) \mid W\!=\!w\big)
    &= \mathrm{Var}_X\!\big(\mathbb{E}_\varepsilon[\mu(X, Y^*) \mid X, W\!=\!w]\big) \notag \\
    &\quad + \mathbb{E}_X\!\big[\mathrm{Var}_\varepsilon\!\big(\mu(X, Y^*) \mid X, W\!=\!w\big)\big].
\end{align}
Now, define the noise-to-signal ratio
\begin{equation}\label{eq:nsr}
  \lambda_w \;\coloneqq\;
  \frac{\mathbb{E}_X\!\bigl[\mathrm{Var}_\varepsilon\!\bigl(\mu(X, Y^*) \mid X, W\!=\!w\bigr)\bigr]}
       {\mathrm{Var}_X\!\bigl(\mathbb{E}_\varepsilon[\mu(X, Y^*) \mid X, W\!=\!w]\bigr)},
\end{equation}
the ratio of the noise to the signal component in \eqref{eq:var_decomp}.
Setting the two expressions equal and solving gives the effective sample size $n_w^{\mathrm{eff}} \coloneqq n_w / (1 + \lambda_w)$, which reduces to $n_w$ in the deterministic case ($\sigma^2_\varepsilon = 0 \Rightarrow \lambda_w = 0$). When $\lambda_w > 0$, stochasticity deflates the effective sample size below $n_w$, inflating the variance of the calibrated estimator. Just like for identification, a natural remedy is to generate $K > 1$ independent draws per unit and use the averaged surrogate $\bar{Y}^*_i = K^{-1}\sum_{k=1}^K Y^*_{i,k}$. This reduces the variance inflation, as the following shows.

\begin{proposition}[Variance reduction via averaging]\label{prop:variance_reduction}
Let
\[
  \mu_K(x, \bar y^*) \;\coloneqq\; \mathbb{E}[Y \mid X=x,\, \bar Y^* = \bar y^*,\, P=0]
\]
and let $\hat\tau_K$ be the estimator \eqref{eq:ate_estimator} with $\mu_K(X_i, \bar Y^*_i)$ in place of $\mu(X_i, Y^*_i)$.
\begin{enumerate}
    \item[(a)] If $Y \perp W \mid X, \bar Y^*$ and $Y \perp P \mid X, \bar Y^*$ with overlap in $(X, \bar Y^*)$, then $\hat\tau_K \xrightarrow{p} \tau$.
    \item[(b)] If, in addition, $\mu_K$ is linear in $\bar y^*$, then $n_w^{\mathrm{eff}}(K) = n_w / (1 + \lambda_w/K) \to n_w$ as $K \to \infty$.
\end{enumerate}
\end{proposition}

\begin{proof}
Part (a) applies Theorem~\ref{thm:identification} with $(\bar Y^*, \mu_K)$ in place of $(Y^*, \mu)$, with the stated conditions being the analogues of Assumptions~\ref{as:surrogacy}--\ref{as:comparability} for the averaged surrogate. For part (b), averaging $K$ conditionally i.i.d.\ draws gives $\mathrm{Var}(\bar Y^*_K \mid X, W) = \sigma^2_\varepsilon(W, X)/K$. Under linearity, the noise term of \eqref{eq:var_decomp} scales by $K^{-1}$ while the signal term is unchanged, so the noise-to-signal ratio at $K$ draws is $\lambda_w/K$.
\end{proof}

Since LLM calls are cheap relative to running human experiments, even modest replication counts ($K = 5$--$10$) substantially reduce the noise-to-signal ratio, bringing the effective sample size close to the ideal $n_w$.

The preceding results assumed that $\mu$ is known. In practice, $\mu$ must be estimated from the experimental sample $P = 0$. Noise in $Y^*$ then introduces a second problem beyond variance inflation, relating to bias. The following proposition quantifies this bias and shows that averaging draws removes it in the limit, under a linear calibration model adopted for tractability.

\begin{proposition}[Attenuation from a noisy surrogate]\label{prop:attenuation}
Suppose that $\mathbb{E}[Y \mid X, \theta(W, X)] = \alpha_0 + \phi_0' X + \beta_0\, \theta(W, X)$ and $\mathbb{E}[\theta \mid X]$ are linear in $X$, that the LLM noise is homoskedastic, $\sigma^2_\varepsilon(W,X) \coloneqq \sigma^2_\varepsilon$, with $\varepsilon \perp Y \mid W, X$, and that $\mu$ is estimated by ordinary least squares of $Y$ on $(X, Y^*)$, $\hat\mu(x, y^*) = \hat\alpha + \hat\phi' x + \hat\beta y^*$. Then
\begin{equation}\label{eq:reliability}
    \hat\beta \xrightarrow{p} R\,\beta_0, \qquad R \coloneqq \frac{\mathrm{Var}(\theta \mid X)}{\mathrm{Var}(\theta \mid X) + \sigma^2_\varepsilon},
\end{equation}
and the calibrated ATE estimator in Eq.~\eqref{eq:ate_estimator} satisfies $\hat\tau \xrightarrow{p} R\,\tau$. Refitting the calibration on the average of $K$ independent draws per unit replaces $\sigma^2_\varepsilon$ by $\sigma^2_\varepsilon / K$ in $R$, so $R_K \to 1$ as $K \to \infty$.
\end{proposition}

\begin{proof}
Regressing $Y$ on $(X, Y^*)$ with $Y^* = \theta + \varepsilon$ is a classical errors-in-variables problem, as $\varepsilon$ is uncorrelated with both $X$ and $\theta(W,X)$ by $\mathbb{E}[\varepsilon \mid W,X] = 0$ and is also uncorrelated with the equation error $Y - \mathbb{E}[Y \mid X, \theta]$ by $\varepsilon \perp Y \mid W,X$. Since $\varepsilon$ is uncorrelated with $X$, the Frisch-Waugh-Lovell theorem reduces the problem to attenuation of $\theta - \mathbb{E}[\theta \mid X]$ by $\varepsilon$, giving $\hat\beta \xrightarrow{p} R\beta_0$ where the reliability ratio uses $\mathrm{Var}(\theta \mid X) = \mathrm{Var}(\theta - \mathbb{E}[\theta \mid X])$, which is the standard reliability-ratio result \citep[][ch.~1]{fuller1987measurement}. Substituting the fitted calibration function into the calibrated estimator then gives $\hat\tau = \hat\phi'(\bar X_1 - \bar X_0) + \hat\beta(\bar Y^*_1 - \bar Y^*_0)$, where randomization of $W$ leads to $\bar X_1 - \bar X_0 \xrightarrow{p} 0$ and $\bar Y^*_1 - \bar Y^*_0 \xrightarrow{p} \Delta_\theta \coloneqq \mathbb{E}[\theta \mid W{=}1] - \mathbb{E}[\theta \mid W{=}0]$. Since the linear model implies $\tau = \beta_0 \Delta_\theta$, it follows that $\hat\tau \xrightarrow{p} R\beta_0 \Delta_\theta = R\tau$. Refitting the calibration function on $\bar Y^*_K$ scales the noise variance to $\sigma^2_\varepsilon/K$, so that $R_K \to 1$ as $K \to \infty$.
\end{proof}

\begin{remark}[Nonlinear $\hat\mu$]\label{rem:attenuation_nonlinear}
When $\hat\mu$ is estimated by flexible nonparametric methods, the analogue of \eqref{eq:reliability} is over-smoothing of the fitted regression toward the marginal mean rather than a scalar coefficient attenuation \citep{fan1993nonparametric}; see also \citet[ch.~12]{carroll2006measurement}. The $K$-averaging solution still applies.
\end{remark}

Taken together, these results make the replication count $K$ a design parameter in LLM-based experiments. A single noisy draw can violate surrogacy and, even when identification holds, add bias and variance to the calibrated estimator. Averaging more draws addresses all three, at the cost of additional LLM calls.

\subsection{Evaluating and testing identification}\label{sec:evaluation}

Identification via Theorem~\ref{thm:identification} rests on surrogacy and comparability, and the preceding section showed that the stochasticity of LLM outputs can make surrogacy harder to satisfy at a single draw. Because neither assumption necessarily holds in any given application, we now develop two empirical tools for assessing them. The first is a moment-condition test that can falsify surrogacy when human outcomes are available on historical treatments, and the second is a sensitivity bound that quantifies the worst-case bias when comparability fails by a lack of overlap between the human and LLM samples. We state both for the single-draw pair $(Y^*, \mu)$ and indicate where the averaged pair $(\bar Y^*_K, \mu_K)$ applies instead, since the diagnostics work identically irrespective of which surrogate the experimenter conditions on.

\subsubsection{Surrogacy falsification test}\label{sec:backtest}

When human outcomes $Y$ are available on historical treatments, Assumption~\ref{as:surrogacy} implies a testable moment condition: for each arm $w \in \{0, 1\}$,
\begin{equation}\label{eq:backtest}
    \mathbb{E}[Y \mid W=w, P=0]
    =
    \mathbb{E}[\mu(X, Y^*) \mid W=w, P=0],
\end{equation}
which follows from iterated expectation applied to $\mu(X, Y^*) = \mathbb{E}[Y \mid X, Y^*, P=0]$, since Surrogacy makes the additional conditioning on $W$ redundant. In practice, \eqref{eq:backtest} can be implemented as a statistical test on the per-arm mean residual $\bar r_w = n_w^{-1}\sum_{i: W_i=w} (Y_i - \hat\mu(X_i, Y^*_i))$, with $\hat\mu$ fit on a separate training fold of $P=0$. By the Lindeberg--L\'{e}vy central limit theorem applied to the held-out residuals, which are i.i.d.\ conditional on the training fold, $\bar r_w / \widehat{\mathrm{SE}}(\bar r_w)$ is asymptotically standard normal under \eqref{eq:backtest}, so the moment is tested by a one-sample $z$-test against zero in each arm. 

Rejection of \eqref{eq:backtest} in such a test then falsifies surrogacy within $P=0$, provided $\hat\mu$ recovers the population regression $\mu(x, y^*) = \mathbb{E}[Y \mid X=x, Y^*=y^*, P=0]$. The diagnostic is therefore best understood as a falsification test, in which rejection refutes the assumption while non-rejection is necessary but not sufficient, since the moment can hold on historical treatments while failing on novel ones. The same test applies to the averaged surrogate when $(\bar Y^*_K, \mu_K)$ replaces $(Y^*, \mu)$, in which case it falsifies the multi-draw surrogacy condition of Proposition~\ref{prop:multidraw} rather than its single-draw counterpart, a distinction that matters because averaging can restore surrogacy that a single noisy draw violates.

\subsubsection{Sensitivity bound for overlap violations}\label{sec:sensitivity}

Assumption~\ref{as:comparability} requires support overlap in $(X,Y^*)$ across the calibration sample ($P=0$) and the artificial sample ($P=1$). We now formalize this requirement and show how deviations from this overlap bound the worst-case discrepancy in the ATE.

Let $Z \coloneqq (X,Y^*)$. For $p\in\{0,1\}$ and $w\in\{0,1\}$, denote by $q_{p,w}(z)$ the density of $Z$ conditional on $(P=p, W=w)$. The arm-specific total variation distance between the calibration and artificial samples \citep[see, e.g.,][pp.\ 83--84]{tsybakov2009nonparametric} is
\begin{equation}\label{eq:TV}
\mathrm{TV}_w \coloneqq \| q_{1,w} - q_{0,w} \|_{TV} = \tfrac{1}{2}\int |q_{1,w}(z) - q_{0,w}(z)|\,dz, \quad w \in \{0, 1\},
\end{equation}
which satisfies $0 \le \mathrm{TV}_w \le 1$, with $\mathrm{TV}_w = 0$ if and only if $q_{0,w} = q_{1,w}$ almost everywhere, and $\mathrm{TV}_w = 1$ if the corresponding supports are disjoint. Equivalently, $1 - \mathrm{TV}_w$ is the \citet{weitzman1970measures} overlap coefficient $\mathrm{OVL}_w = \int \min\{q_{0,w}(z), q_{1,w}(z)\}\,dz$, which measures the fraction of mass shared between $q_{0,w}$ and $q_{1,w}$ within arm $w$. We impose the following condition.

\begin{assumption}[Bounded outcomes]\label{as:bounded}
    There exists a constant $B<\infty$ such that $|Y|\le B$ almost surely.
\end{assumption}

Under this assumption, the calibration function $\mu(z) \coloneqq \mathbb{E}[Y\mid Z=z,P=0]$ is also bounded by $B$.

The following proposition bounds the discrepancy between the LLM-based ATE computed using the calibration function and the corresponding quantity computed in the experimental sample.

\begin{proposition}[Sensitivity bound]\label{prop:sensitivity}
Under Assumption~\ref{as:bounded}, define
$\Delta_p \coloneqq \mathbb{E}[\mu(Z)\mid P=p,W=1]-\mathbb{E}[\mu(Z)\mid P=p,W=0]$ for $p\in\{0,1\}$. Then
\begin{equation}\label{eq:sensitivity}
\big|\Delta_1 - \Delta_0\big| \le 2B\big(\mathrm{TV}_0 + \mathrm{TV}_1\big).
\end{equation}
\end{proposition}

\begin{proof}
Fix $w\in\{0,1\}$ and let $m_{p,w}\coloneqq\mathbb{E}[\mu(Z)\mid P=p,W=w]$, so that $\Delta_p = m_{p,1} - m_{p,0}$. By definition of $q_{p,w}$, we have $m_{1,w}-m_{0,w} = \int \mu(z)\big(q_{1,w}(z)-q_{0,w}(z)\big)\,dz$. Since $|\mu(z)|\le B$ by Assumption~\ref{as:bounded},
\begin{align*}
|m_{1,w}-m_{0,w}| \le B\int |q_{1,w}(z)-q_{0,w}(z)|\,dz = 2B \cdot \mathrm{TV}_w,
\end{align*}
where the last equality follows from the total variation identity in Equation~\eqref{eq:TV}. The triangle inequality applied to $\Delta_1 - \Delta_0 = (m_{1,1} - m_{0,1}) - (m_{1,0} - m_{0,0})$ then gives $|\Delta_1 - \Delta_0| \le 2B(\mathrm{TV}_0 + \mathrm{TV}_1)$.
\end{proof}

Proposition~\ref{prop:sensitivity} bounds the worst-case discrepancy in the ATE identified when the overlap condition in Assumption~\ref{as:comparability} fails, and can be interpreted as a partial identification result in the sense of \citet{manski2003partial}. The bound increases linearly in the arm-specific total variation distances: It vanishes when the distributions of $(X, Y^*)$ coincide across samples within each arm, and it can be at most $4B$, which occurs when the supports are disjoint in both arms.

Sharper bounds are possible under additional assumptions, for instance by replacing the total variation distances with a Wasserstein metric scaled by the Lipschitz constant of $\mu$ and assuming smoothness, or by exploiting estimable density ratios between the two samples. Such bounds come at the cost of restricting the function class by requiring additional regularity conditions on $\mu$. The distribution-free bound presented here can therefore be viewed as conservative and generally applicable.

\subsubsection{Estimating the bound}\label{sec:estimating_bound}

The sensitivity bound from Proposition~\ref{prop:sensitivity} on $|\Delta_1 - \Delta_0|$ depends only on the outcome bound $B$ and the arm-specific total variation distances $\mathrm{TV}_w$, which are functions of the joint distributions of $(X,Y^*)$ in the experimental and artificial samples. In practice, $B$ is known if the outcome is naturally bounded (e.g., binary or scaled outcomes) and may otherwise be approximated using domain knowledge or historical data.

The total variation distances can be estimated using plug-in methods. One approach is to estimate the conditional densities $q_{p,w}(z)$ within each $(P=p,W=w)$ cell using flexible parametric or nonparametric models, and then compute
\begin{equation}
    2B\big(\widehat{\mathrm{TV}}_0 + \widehat{\mathrm{TV}}_1\big),
\end{equation}
where
\begin{equation}
\widehat{\mathrm{TV}}_w
=
\tfrac{1}{2} \int |\hat q_{1,w}(z) - \hat q_{0,w}(z)|\,dz.
\end{equation}
Alternatively, $\mathrm{TV}_w$ may be approximated via density-ratio estimation or classification-based approaches that distinguish $(P=0)$ from $(P=1)$ within each arm. The integrals can be evaluated numerically using Monte Carlo methods.

When $Y^*$ is stochastic, $\mathrm{TV}_w$ is computed on the full distribution of $(X, Y^*)$ including the LLM sampling noise, so holding the decoding configuration $\mathcal{D}$ fixed across the calibration and artificial samples leaves the noise component symmetric across both arms and absorbs it into the bound identically. Varying $\mathcal{D}$ across samples, for instance through a different temperature during deployment, introduces an additional source of distributional shift that inflates $\mathrm{TV}_w$ and widens the bound. Since $\mu_K$ is also bounded by $B$, the same bound applies when we draw $K > 1$ samples per unit and replace $(X, Y^*)$ with $(X, \bar Y^*_K)$ and $\mu$ with $\mu_K$. The replication count $K$ changes only the surrogate on which $\mathrm{TV}_w$ is measured, not whether the bound holds.

\section{Monte Carlo Experiments}\label{sec:simulation}

We illustrate the identification result, the implications of violations of the assumptions, and the role of stochastic and multiple LLM draws using a set of Monte Carlo experiments. Each design draws $1{,}000$ observations per arm, giving $n = 2{,}000$ in the two-arm designs and $n = 3{,}000$ in the three-arm surrogacy-falsification test scenario. This choice can be interpreted as a medium-sample setting in which human samples are costly to collect. We use $1{,}000$ Monte Carlo replications, reduced to $200$ for the overlap-based scenarios because the kernel-density estimation they require is expensive. We use a linear data-generating process (DGP) where the true calibration function is linear in $(Y^*, X)$, as well as a nonlinear DGP in which $\mu(x, y^*) = 0.3 (y^*)^2 + 0.5 \cos(\pi x)$. For the linear DGP, we fit $\hat\mu$ by ordinary least squares; for the nonlinear DGP, we use a random forest. Full code and seeds are available in the supplementary material.

\subsection{Identification under correct assumptions}

We first verify Theorem~\ref{thm:identification} when Assumptions~\ref{as:surrogacy}--\ref{as:comparability} hold. Figure~\ref{fig:identification} plots the sampling distribution of the calibrated ATE alongside the raw LLM ATE, for both the linear and nonlinear DGPs. The raw LLM estimator is severely biased (mean $\approx 0.75$ against a true $\tau = 0.30$ for the linear DGP; $\approx 0.50$ against $\tau \approx 0.27$ for the nonlinear DGP), reflecting the fact that $Y^*$ and $Y$ live on different scales. Applying the calibration function recovers the human ATE almost exactly, with residual bias of order $10^{-3}$ in both DGPs and RMSE of $0.012$ (linear) and $0.021$ (nonlinear). The deterministic DGP ($\sigma^2_\varepsilon = 0$) yields essentially identical results, consistent with Remark~\ref{rem:deterministic}.

\begin{figure}[t]
\centering
\includegraphics[width=\linewidth]{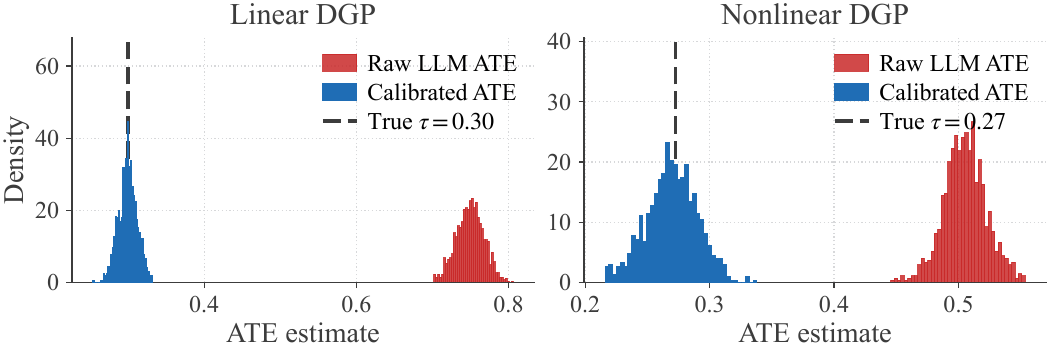}
\caption{Sampling distribution of the calibrated ATE (blue) versus the raw LLM ATE (orange) under Assumptions~\ref{as:surrogacy}--\ref{as:comparability}. Calibration recovers the true human ATE $\tau$ (dashed) in both the linear (left) and nonlinear (right) DGPs; the raw LLM ATE is systematically biased.}
\label{fig:identification}
\end{figure}

We further verify the parametric $\sqrt{n}$-rate implied by Theorem~\ref{thm:identification} by sweeping the per-population sample size $n$ on the linear DGP. Figure~\ref{fig:consistency_rate} reports the empirical RMSE of the calibrated ATE across $n \in \{250, 500, 1{,}000, 2{,}000, 5{,}000, 10{,}000, 20{,}000\}$. Panel (a) plots RMSE against $n$ on log--log axes alongside the $n^{-1/2}$ reference. We see that the empirical curve tracks the reference closely across the full grid. Panel (b) plots $\mathrm{RMSE} \cdot \sqrt{n}$ against $n$, where a flat line in $n$ corresponds to the parametric $\sqrt{n}$-rate. The empirical $\mathrm{RMSE} \cdot \sqrt{n}$ hovers near $0.55$ across two orders of magnitude in $n$, with deviations consistent with Monte Carlo sampling noise at the replication count used here.

\begin{figure}[t]
\centering
\includegraphics[width=\linewidth]{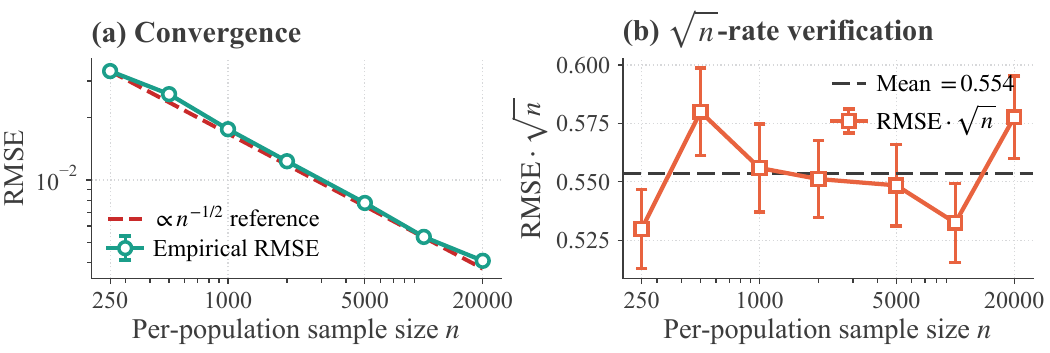}
\caption{Empirical $\sqrt{n}$-consistency of the calibrated ATE on the LinearDGP. (a) RMSE against the per-population sample size $n$, with the $\propto n^{-1/2}$ reference for comparison. (b) $\mathrm{RMSE} \cdot \sqrt{n}$ against $n$, hovering near $0.55$ across the full grid and consistent with the parametric $\sqrt{n}$-rate. Error bars denote Monte Carlo standard errors.}
\label{fig:consistency_rate}
\end{figure}

We also test whether the results are sensitive to the choice of calibration model. Table~\ref{tab:robustness} reports the estimated ATEs from OLS, random forest, and gradient-boosted trees on both DGPs. The model best matched to the DGP attains the lowest RMSE, with OLS best on the linear DGP and the nonparametric methods best on the nonlinear one, but a mismatched model is not much worse, as OLS on the nonlinear DGP still recovers the ATE to within about $6\%$ relative bias. This error is small relative to the bias from violating Surrogacy or Comparability (Figure~\ref{fig:failure_modes}), suggesting identification matters more than specification.

\begin{table}[!htb]
\centering
\caption{Sensitivity of the calibrated ATE to the choice of calibration model. Cells report summary statistics over $1{,}000$ Monte Carlo replications with $n=2{,}000$. Random forests and gradient-boosted trees use default hyperparameters with 200 trees.}
\label{tab:robustness}
\resizebox{\columnwidth}{!}{\begin{tabular}{llrrrr}
\toprule
DGP & Calibration model & Mean $\hat\tau$ & Bias & RMSE & SD \\
\midrule
Linear (true $\tau = 0.300$) & OLS & 0.2994 & -0.0006 & 0.0122 & 0.0122 \\
 & Random forest & 0.2996 & -0.0004 & 0.0131 & 0.0131 \\
 & Gradient boosting & 0.2992 & -0.0008 & 0.0134 & 0.0134 \\
\midrule
Nonlinear (true $\tau = 0.273$) & OLS & 0.2552 & -0.0175 & 0.0264 & 0.0197 \\
 & Random forest & 0.2695 & -0.0032 & 0.0207 & 0.0205 \\
 & Gradient boosting & 0.2704 & -0.0023 & 0.0206 & 0.0205 \\
\bottomrule
\end{tabular}
}
\end{table}

\subsection{Bias under identification violations}

Figure~\ref{fig:failure_modes} shows what happens when the identifying assumptions fail. In panel (a), we introduce a direct effect of $W$ on $Y$ that bypasses $Y^*$, governed by a parameter $\gamma$, violating Assumption~\ref{as:surrogacy}. We do so by adding a term $\gamma W$ to the outcome equation, so that $Y = \mu(X, Y^*) + \gamma W + \eta$ and a portion of the total treatment's effect on $Y$ is direct through a channel that $Y^*$ does not capture, with $\gamma = 0$ recovering the surrogacy-respecting DGP where the treatment effect is only indirect. We find that the bias of the calibrated estimator then scales linearly in $\gamma$. The slope, however, is smaller in magnitude than the one-for-one rate one might naively expect. The reason is that the calibration function is, by construction, a regression of $Y$ on $(X, Y^*)$ that excludes $W$, so the omitted $\gamma W$ is partly absorbed into the fitted coefficient on $Y^*$ through the correlation between $W$ and $Y^*$, and this coefficient is in turn attenuated by the measurement noise in $Y^*$. A closed-form calculation for the linear DGP delivers a slope of $-0.395$, which the simulation recovers at $-0.396$ across the $\gamma$ grid. In panel (b), we induce a shift of magnitude $\delta$ in the slope of the calibration function between $P=0$ and $P=1$, in turn violating Assumption~\ref{as:comparability}. The bias again scales linearly in $\delta$, with a larger slope because the violation compounds across arms. Both panels show that neither assumption is redundant, with surrogacy and comparability failing in empirically distinct ways that produce predictable bias in the calibrated estimator.

\begin{figure}[t]
\centering
\includegraphics[width=\linewidth]{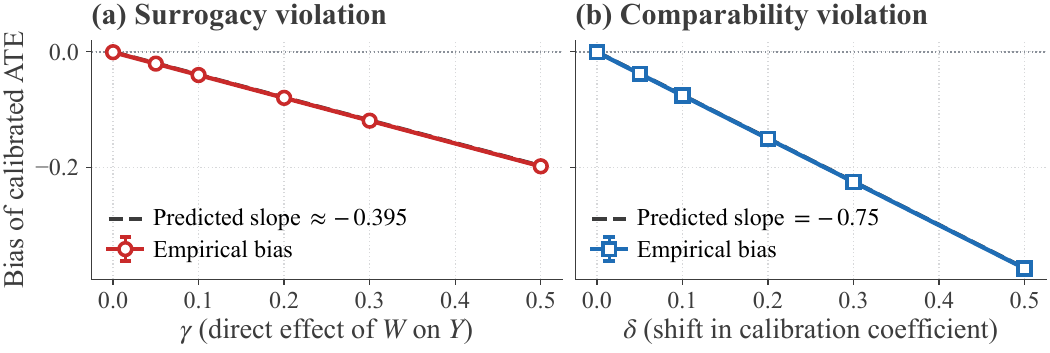}
\caption{Bias of the calibrated ATE under violations of each identifying assumption. (a) Surrogacy is violated by a direct effect $\gamma$ of $W$ on $Y$. (b) Comparability is violated by a shift $\delta$ in the calibration slope between $P=0$ and $P=1$. In both cases, bias scales linearly in the violation parameter. Dashed lines show the closed-form theoretical slopes from the LinearDGP and OLS calibration, $-0.395$ for $\gamma$ in panel (a) and $-0.75$ for $\delta$ in panel (b). Error bars denote Monte Carlo standard errors.}
\label{fig:failure_modes}
\end{figure}

\subsection{Multi-draw surrogacy and estimation quality}

We operationalize the multi-draw setting with a deliberately adversarial DGP in which $Y$ depends on the latent mean $\theta(W, X) = \mathbb{E}[Y^* \mid W, X]$ rather than on any realized $Y^*$, and the per-draw noise is large ($\sigma_\varepsilon = 1.5$), so that single-draw surrogacy fails by construction rather than as a generic property of the method. Figure~\ref{fig:multidraw} reports the mean calibrated ATE across $K$. It is strongly attenuated toward zero at $K=1$ but approximately recovers the true ATE by $K=200$ (mean $\approx 0.02$ versus $\approx 0.28$ against the true $\tau = 0.30$). The large $K$ required here reflects the deliberately high noise level ($\sigma_\varepsilon = 1.5$) chosen for illustrative purposes; fewer draws may suffice in practice, as discussed in Section~\ref{sec:estimation_stochastic} and empirically shown in Section~\ref{sec:empirical}. This result confirms both predictions of our theory, with the de-attenuation tracking the reliability ratio $R_K$ of Proposition~\ref{prop:attenuation} and surrogacy restored in the large-$K$ limit by Proposition~\ref{prop:multidraw}.

\begin{figure}[t]
\centering
\includegraphics[width=0.72\linewidth]{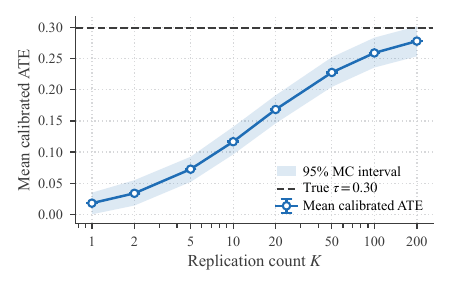}
\caption{Multi-draw surrogacy relaxation. When $Y$ depends on the latent $\theta(W,X)$ rather than on a particular realization, single-draw surrogacy fails. Averaging $K$ independent draws of $Y^*$ per unit restores identification as $K$ grows, consistent with the multi-draw surrogacy condition of Proposition~\ref{prop:multidraw}. Error bars denote Monte Carlo standard errors.}
\label{fig:multidraw}
\end{figure}

Proposition~\ref{prop:variance_reduction} states that the noise component of the calibrated estimator's variance shrinks as $1/K$, while Proposition~\ref{prop:attenuation} states that the attenuation from a stochastic surrogate shrinks through the reliability ratio $R_K \to 1$. Figure~\ref{fig:estimation_quality} reproduces both predictions, with RMSE falling an order of magnitude between $K=1$ and $K=50$ (panel a) and the mean estimate rising monotonically toward the true $\tau$ (panel b). The curvature in panel (b) tracks the reliability ratio, the same pattern that classical measurement-error theory produces for a linear specification with a noisy regressor \citep{fuller1987measurement}.

\begin{figure}[t]
\centering
\includegraphics[width=\linewidth]{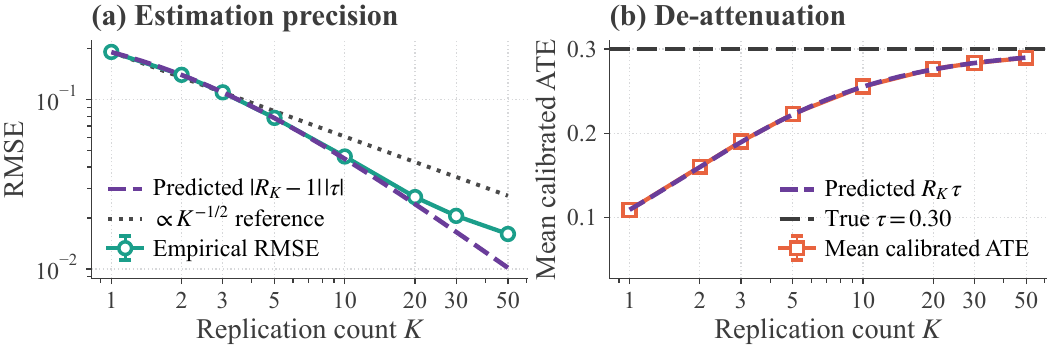}
\caption{Estimation quality as a function of the replication count $K$. (a) RMSE falls an order of magnitude between $K=1$ and $K=50$, matching the $1/K$ scaling of the noise term in Proposition~\ref{prop:variance_reduction}(b). (b) The mean calibrated ATE approaches the true $\tau$ as $K$ grows, tracking the reliability ratio $R_K$ (Proposition~\ref{prop:attenuation}).}
\label{fig:estimation_quality}
\end{figure}

\subsection{Overlap sensitivity and surrogacy falsification}\label{sec:simulation_sensitivity}

To stress-test Proposition~\ref{prop:sensitivity}, we introduce an asymmetric location shift of size $s$ to $Y^*$ in the treated arm of $P=1$ only, so that the $W=1$ arm extrapolates beyond the $P=0$ support while the $W=0$ arm remains aligned. Under the linear calibration used here, this produces $|\Delta_1 - \Delta_0| = |b_1 \cdot s|$ exactly, where $b_1$ is the coefficient on $Y^*$ in the linear calibration function, while inflating $\mathrm{TV}_1$ and leaving $\mathrm{TV}_0 \approx 0$. Figure~\ref{fig:sensitivity} plots both quantities. The actual $|\Delta_1 - \Delta_0|$ grows linearly with $s$, and the theoretical bound $2B(\mathrm{TV}_0 + \mathrm{TV}_1)$ (with $B=3$) lies strictly above it at every shift, saturating at $2B$ as $\mathrm{TV}_1 \to 1$ in the shifted arm alone. The bound holds with probability $1$ across every replication and every shift in our grid, empirically validating Proposition~\ref{prop:sensitivity}. The gap between $|\Delta_1 - \Delta_0|$ and the bound reflects the conservativeness of the worst-case analysis. The smooth linear $\mu$ used here never approaches the extreme values $\pm B$ that would make the bound tight. Quantitatively, the empirical-to-theoretical bound ratio $|\Delta_1 - \Delta_0| / [2B(\mathrm{TV}_0 + \mathrm{TV}_1)]$ grows from $0.015$ at zero shift to $0.125$ at the largest shift in our grid, so the bound is never tighter than a factor of eight even when the LLM in the treated arm is shifted by two units of $Y^*$, which, for reference, is roughly $3.6$ times its standard deviation.

\begin{figure}[t]
\centering
\includegraphics[width=0.72\linewidth]{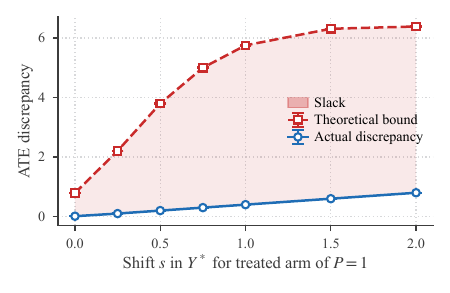}
\caption{Empirical validation of the sensitivity bound. An asymmetric shift is applied to $Y^*$ in the treated arm of $P=1$, which inflates $\mathrm{TV}_1$ only in that arm. The actual $|\Delta_1 - \Delta_0|$ (blue) grows linearly. The theoretical bound from Proposition~\ref{prop:sensitivity}, $2B(\mathrm{TV}_0 + \mathrm{TV}_1)$ (red dashed), grows and saturates at $2B$ as $\mathrm{TV}_1 \to 1$. The bound holds with probability one across all shifts and replications.}
\label{fig:sensitivity}
\end{figure}

As a final check, we illustrate the necessary-but-not-sufficient nature of the falsification test from Section~\ref{sec:backtest}. We learn a calibration function on two historical treatments ($W=0, W=1$), where surrogacy holds, and then apply it to a novel treatment ($W=2$) with a direct effect not mediated by $Y^*$. The calibrated ATE for the historical contrast recovers the human ATE up to a residual bias of $+0.12$ (on a true $\tau = 0.30$), while for the novel contrast it misses the human ATE by $-0.14$ (on a true $\tau = 0.72$). This illustrates that the test is necessary but not sufficient: the moment condition can hold approximately on historical treatments while identification fails for novel ones.

\section{Empirical Application: Upworthy Headline A/B Tests}\label{sec:empirical}

We use the Upworthy Research Archive \citep{matias2021upworthy}, an open dataset of $32{,}487$ randomized headline A/B tests that Upworthy ran between January 2013 and April 2015, released with a data descriptor in \emph{Nature Scientific Data}. In each test, visitors were randomly assigned to one of several headline variants of an otherwise identical article, and the platform recorded impressions and clicks for each variant, which gives the randomized, real-human treatment effects we need to benchmark a calibrated LLM surrogate.

From the archive's open access release we first drop the tests flagged for a randomization bug (see the authors' 2024 erratum for details). We also omit variants with fewer than $100$ impressions to ensure our results are not overly driven by noise, and pair the two highest-impression variants of each test into a single binary A/B contrast. This leaves us with $3{,}603$ paired tests with a median of $3{,}178$ impressions per variant.

To test whether a calibrated LLM surrogate recovers a \emph{human} treatment effect, the treatment--control contrast should be assigned independently of the LLM. To this end, we take the treatment to be a feature of the headline itself, and define a treated variant as one whose headline is phrased as a question. We restrict our analysis to the $417$ tests whose two variants differ on this feature, so that within each test visitors are randomized between headlines that do or do not pose a question. The estimated ATE on their click-through rates is then a genuine human causal effect that the calibrated surrogate should recover, where the contrast represents the effect of a question-headline as a whole.

\subsection{Variables and surrogate construction}

The unit of analysis is a test variant aggregated over user interactions, with outcome $Y \in [0,1]$ the click-through rate and treatment $W$ the question indicator defined above. We use headline length, the test's calendar week, and the number of variants in the test, as covariates $X$ for fitting the calibration function, where the latter two are fixed before randomization per the data. We use \texttt{gpt-4o-mini} at temperature $0.7$ to construct the LLM surrogates $Y^* \in [0,1]$, prompting the model to predict the click-through rate $Y$ of a typical Upworthy-era Facebook user, as follows:

\medskip
\begin{quote}\small\ttfamily
\textbf{System.} You estimate click-through rates (CTRs) for headlines posted on
Upworthy.com between 2013 and 2015. Upworthy's traffic came mostly from the
Facebook news feed in that period. Given a candidate headline, predict the CTR
that a typical Upworthy-era Facebook user would exhibit when shown it in their
feed. Respond with a single integer giving the predicted CTR in basis points
(1 bp = 0.01\%). For example, an answer of `80' means 0.80\%. Do not include any
other text, units, or commentary. Only the integer. \\
\medskip
\textbf{User.} Headline: \{headline\} Predicted CTR (basis points):
\end{quote}
\medskip

We draw $K=10$ independent LLM responses per prompt and convert each integer basis-point answer to a probability $Y^* \in [0,1]$ by dividing by $10{,}000$. For example, a response of $80$ becomes an LLM-predicted CTR of $Y^* = 0.008$.

We use a diagnostic to check whether the LLM has memorized this dataset, which would let it recall the realized CTRs instead of predicting them, but find no evidence of it. See Section~\ref{sec:uw_diagnostics} for details.

\subsection{Main results}\label{sec:uw_question}

Table \ref{tab:uw_question} shows the ATE estimates on the actual CTR and those from the surrogates before and after calibration. We construct each calibrated prediction by five-fold cross-fitting over tests, fitting $\hat\mu$ on the out-of-fold variants of the full sample and applying it to the held-out fold, so that no test informs its own calibration while the calibration still draws on the entire dataset for power. The raw surrogate recovers the sign of the human effect but attenuates its magnitude to roughly two-fifths ($-4.5 \times 10^{-4}$, SE $1.2 \times 10^{-4}$), and its average gap to the human contrast is significantly positive ($+7.2 \times 10^{-4}$, $t = 2.7$), so the raw surrogate is biased toward zero rather than recovering the effect. Linear calibration reduces this gap but remains significantly attenuated ($t = 2.3$), whereas the two nonparametric calibration estimators (RF and GBT) close it to within sampling error, with gaps of $+3.7 \times 10^{-4}$ ($t = 1.2$) for the random forest and $+4.4 \times 10^{-4}$ ($t = 1.4$) for the gradient-boosted trees. Here, the attenuation of the raw surrogate can be understood from the measurement-error attenuation that Proposition~\ref{prop:attenuation} predicts for a noisy surrogate. The fact that only the nonparametric estimators have nonsignificant errors reflects that a linear function may not fully capture the surrogate relationship due to misspecified functional form, even if it includes all predictors.

\begin{table}[t]
\centering
\caption{Estimated ATEs $\hat\tau$. Each ATE is the within-test paired difference between the question and non-question variant, averaged over tests, with standard errors from the test-level paired differences. Calibrated predictions use five-fold cross-fitting over tests.}
\label{tab:uw_question}
\begin{tabular}{lccc}
\toprule
Estimator & $\hat\tau$ & SE & Gap ($t$-statistic) \\
\midrule
Human & $-0.00116$ & $0.00024$ & --- \\
Raw LLM, $K=10$ & $-0.00045$ & $0.00012$ & $+0.00072$ ($2.7$) \\
Calibrated LLM, OLS & $-0.00056$ & $0.00013$ & $+0.00061$ ($2.3$) \\
Calibrated LLM, RF & $-0.00079$ & $0.00022$ & $+0.00037$ ($1.2$) \\
Calibrated LLM, GBT & $-0.00072$ & $0.00024$ & $+0.00044$ ($1.4$) \\
\bottomrule
\end{tabular}
\end{table}

\begin{figure}[tbp]
\centering
\includegraphics[width=\linewidth]{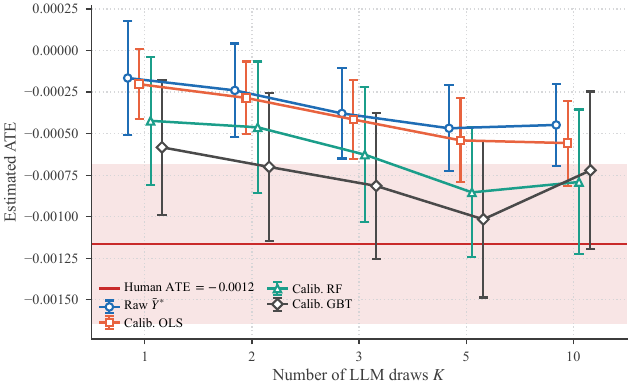}
\caption{Estimated ATEs for the raw surrogate $\bar Y^*$ and the three calibrated estimators as a function of the number of averaged LLM draws $K$, with $95\%$ confidence intervals. The red line and shaded band mark the human ATE and its $95\%$ confidence interval.}
\label{fig:uw_ksweep}
\end{figure}

Figure~\ref{fig:uw_ksweep} visualizes each estimator's ATE as a function of the number of LLM draws per unit-prompt, that we then use as averaged surrogate $\bar Y^*$. Increasing $K$ from 1 to 10 reduces the per-draw sampling noise in $\bar Y^*$ and raises its reliability ratio $R_K$. It thereby de-attenuates the estimates toward the human ATE as $K$ grows. Most of the change occurs over the first few draws. This finding is consistent with Propositions~\ref{prop:variance_reduction} and~\ref{prop:attenuation} by showing that noise decreases and that estimates de-attenuate toward the human ATE as $K$ grows. As our theoretical results predict, the two nonparametric calibrations exhibit least error, reaching the human two-standard-error band by $K \approx 5$ (at $K \approx 5$ the gradient-boosted estimate is $-1.0 \times 10^{-3}$, against a human effect of $-1.2 \times 10^{-3}$), whereas the raw surrogate and the linear calibration de-attenuate more slowly and remain in excess of 2 standard errors away even at $K=10$.

\subsection{Diagnostics}\label{sec:uw_diagnostics}

We run several diagnostics on the results, all on the same sample of tests as the main analysis. The first two stress-test whether the calibration function $\hat\mu$ transports across populations under a deliberately induced distribution shift. We induce the shift within the human sample rather than test comparability directly across $P=0$ and $P=1$, since the artificial sample carries no human outcome against which $\hat\mu$ could be checked. Concretely, we order the question tests by date and split them into an early calibration set (the first $80\%$, $n_0 = 666$ variants across $333$ tests) and a late hold-out (the final $20\%$, $n_1 = 168$ variants across $84$ tests), fitting $\hat\mu$ on the early set and evaluating it on the late hold-out. The latter two diagnostics instead probe the surrogate directly, checking whether the estimator recovers a known effect when one is constructed and whether the LLM's surrogate is an actual prediction and not just memorized from its training data.

\paragraph{Surrogacy falsification test.} Section~\ref{sec:backtest} shows that under Assumption~\ref{as:surrogacy} the moment $\mathbb{E}[Y - \mu(X, Y^*)] = 0$ holds on a held-out sample. We find that the realized moment depends on the calibration model: OLS gives $+2.0 \times 10^{-3}$ (bootstrap standard error $5.3 \times 10^{-4}$, $3.8$ SEs from zero), random forests $-7.9 \times 10^{-4}$ ($5.0 \times 10^{-4}$, $1.6$ SEs), and gradient-boosted trees $+7.0 \times 10^{-4}$ ($4.9 \times 10^{-4}$, $1.4$ SEs). The test therefore falsifies the implication of Theorem~\ref{thm:identification} for OLS but not for the two nonparametric calibrations. This result is consistent with Table~\ref{tab:uw_question}, where OLS remains significantly attenuated in the ATE while the random forest and gradient-boosted trees recover the human ATE within sampling error. Note, however, that the moment is a pooled level check on $\hat\mu$ rather than on the contrast $\hat\tau$, so the two results need not coincide in general, but here the calibration that fails the level test is also the one that misses the ATE.

\paragraph{Overlap and sensitivity bound.} The temporal split induces distribution shift whereby the arm-specific total variation distances between the early and late test sets are $\widehat{\mathrm{TV}}_0 = 0.828$ and $\widehat{\mathrm{TV}}_1 = 0.821$ (equivalently, only $\widehat{\mathrm{OVL}}_w \approx 0.18$ of the joint mass overlaps within either arm). This suggests that Assumption~\ref{as:comparability} holds only partially on the joint $(X, Y^*)$ distribution. Computing the bound from Proposition~\ref{prop:sensitivity} with $B = \max_{y \in \mathcal{Y}} |y| = 1$ yields a worst-case bound of $2(\widehat{\mathrm{TV}}_0 + \widehat{\mathrm{TV}}_1) \approx 3.30$ on the error between any LLM-based calibrated estimate and the human ATE. By contrasting this bound with the realized error, we can see whether the weak overlap translates to a large error in this particular application. We find that it does not. The realized error is at most $\approx 1.3 \times 10^{-4}$ across the three calibration methods, an empirical-to-theoretical ratio of at most $\approx 3.9 \times 10^{-5}$, more than four orders of magnitude below the worst-case ratio of one, which lends credibility to the estimate despite the incomplete overlap. The bound is loose for two reasons. First, it is a worst case over all calibration functions bounded by $B$, while our estimated $\hat\mu$ is far smoother. Second, the total variation distance counts every distributional difference between the early and late periods, yet much of that difference is a common shift in the overall CTR level over time. Such a shift moves $\hat\mu$ by about the same amount in both arms and cancels in the contrast $\hat\tau$, so it inflates $\widehat{\mathrm{TV}}_w$ without adding error to the ATE.

\paragraph{Positive control.} To check that the estimator recovers an ATE when one exists, we set $\tilde Y_i = Y_i + \tau^\dagger W_i$ and construct an informative surrogate $\tilde Y^*_{i,k} = \tilde Y_i + \nu_{i,k}$ with $\nu_{i,k} \sim \mathcal{N}(0, \sigma_\nu^2)$ over $K = 10$ draws, holding the covariates $X$, the temporal split, and the grouping of variants into tests fixed. Figure~\ref{fig:uw_synth} reports $\hat\tau$ across values of $\tau^\dagger$ and noise regimes with reliability ratios $R_{K=10} \in \{0.91, 0.53, 0.09\}$. At the null of $\tau^\dagger = 0$, every estimate lies within $6 \times 10^{-4}$ of zero across all methods and noise regimes. At nonzero $\tau^\dagger$, the estimator tracks the constructed ATE closely when reliability is high and attenuates toward zero as $R_K$ falls (the fraction of effect size recovered is around $0.5$ at $\sigma_\nu = 3\sigma_Y$ and near zero at $\sigma_\nu = 10\sigma_Y$), consistent with Proposition~\ref{prop:attenuation}.

\begin{figure}[tbp]
\centering
\includegraphics[width=0.62\linewidth]{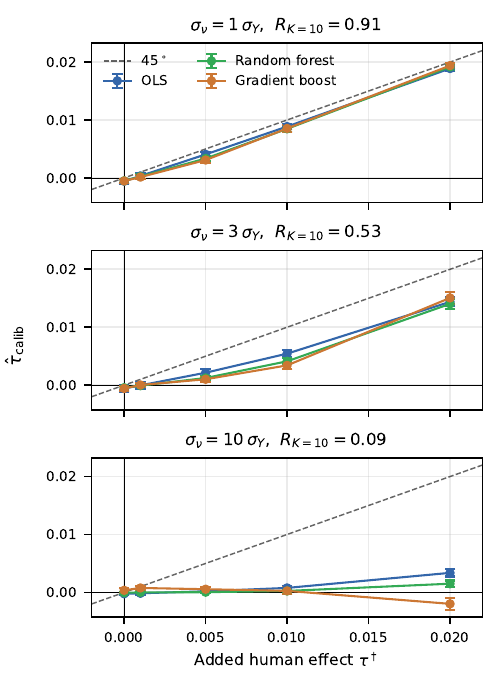}
\caption{Positive control on the Upworthy data. A known effect $\tau^\dagger$ is added to the human outcome ($\tilde Y_i = Y_i + \tau^\dagger W_i$), and an informative surrogate is constructed as $\tilde Y^*_{i,k} = \tilde Y_i + \nu_{i,k}$ with $\nu_{i,k} \sim \mathcal{N}(0, \sigma_\nu^2)$ across $K = 10$ draws. Each panel corresponds to one noise regime; the reliability ratio $R_{K = 10}$ from Proposition~\ref{prop:attenuation} controls how closely each calibration method tracks the $45^\circ$ line. Error bars are one cluster-bootstrap standard error at \texttt{test\_id}.}
\label{fig:uw_synth}
\end{figure}

\paragraph{Did the LLM memorize the data?} A possible explanation for the predictive content of the LLM surrogates is that \texttt{gpt-4o-mini} has memorized the precise 2013--2015 Upworthy headlines from its pre-training corpus and is simply recalling them. To test this, we take a random sample of $300$ headlines, present the LLM with only the first four words of each headline, ask it to complete the headline, and then score the output against the true suffix. We use the token-level $F_1$ to measure the word overlap between the completion and the true suffix, where a high value indicates that the model reproduced the actual headline rather than inventing a plausible one. We find a mean token-level $F_1$ of $0.093$ (median $0.091$), with no output exceeding $F_1 = 0.5$. For reference, \citet{carlini2023quantifying} report an $F_1$ range of $0.1$--$0.3$ for non-memorized text. Our values sit at or below the lower end of this range, implying little evidence of memorization and that the LLM surrogates are indeed genuine predictions. Figure~\ref{fig:uw_contam} plots the full token-$F_1$ distribution.

\begin{figure}[tbp]
\centering
\includegraphics[width=0.72\linewidth]{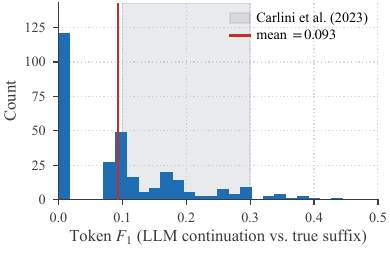}
\caption{Distribution of token-$F_1$ for the headline-completion test. For $300$ randomly sampled Upworthy headlines, we present \texttt{gpt-4o-mini} with the first four words and score its continuation against the true suffix by token $F_1$. The distribution is concentrated below $F_1 = 0.2$, and no completion exceeds $F_1 = 0.5$, indicating little evidence that the model has memorized these headlines. The red line marks the mean ($0.093$) and the shaded band the $0.1$--$0.3$ non-memorized range reported by \citet{carlini2023quantifying}.}
\label{fig:uw_contam}
\end{figure}

\section{Design and Deployment Considerations for LLM Surrogates}\label{sec:design}

Using LLMs to generate surrogates introduces several novel design and deployment considerations that affect whether they can identify human ATEs. We now discuss the role of these considerations and provide some high-level guidance based on our theoretical results and findings from the literature.

\subsection{Model training}\label{sec:model_training}

How the LLM is built and trained determines whether surrogacy (Assumption~\ref{as:surrogacy}) can hold. Consider three representative scenarios:

\begin{enumerate}
    \item \textbf{Foundational model case.} The model relies only on generic knowledge of user behavior, generating artificial responses from an abstract understanding of how the content $X$ and the treatment $W$ map to the proxied outcome $Y^*$. Assumption~\ref{as:surrogacy} is unlikely to hold unless the LLM's generic knowledge ac\-cu\-rate\-ly captures both baseline user behavior patterns and the treatment effect of interest.

    \item \textbf{Product-specific training without treatments.} An alternative approach is to train the LLM on user behavior data from a particular setting, such as the specific product, website or app in question, but in the absence of any treatments or without encoding treatment assignments in the data. Practically, this means a company fine-tunes the LLM on historical logs of data, for instance from the production experience served to users. This approach will tend to produce LLM responses that better represent user patterns, but the lack of explicit or encoded treatment variation implies that the LLM will not be encouraged to learn treatment effects. In other words, $Y^*$ is expected to be more aligned with $Y$ than under (1), but not with $Y(w)$. Hence, Assumption~\ref{as:surrogacy} is still likely to be violated in a particular application.

    \item \textbf{Product-specific training with treatments.} Finally, we can train an LLM on data from specific settings with treatment variation and treatment assignment encoded. For instance, we can fine-tune an LLM separately on user data from a control group and a treatment group, using identical prompting and metrics for measuring alignment (e.g., a distributional divergence measure; c.f.\ Section~\ref{sec:perfect}). In this case, the model has a greater chance of satisfying the surrogacy assumption, particularly when the new intervention for which to estimate the ATE is similar to treatments encoded in the data used to fine-tune the LLM.

\end{enumerate}

\subsection{Prompting}\label{sec:prompt}

So far, we have treated the prompt $I$ as fixed. In practice it must be specified, and prompt choice may have a substantial impact on whether the identifying assumptions hold, as LLM outputs are well documented to be sensitive to prompt wording \citep[e.g.,][]{gao2025scylla}. For a fixed LLM $M=m$, there may exist prompts under which the identifying assumptions are approximately satisfied, and others under which they fail.

Prompt choice can therefore be viewed as selecting a data-gen\-er\-at\-ing mechanism for $Y^*$, directly affecting whether Surrogacy and Comparability hold. Prompt optimization and its consequences should therefore be part of the process of A/B testing with LLMs.

\subsection{Temperature}\label{sec:temperature}

The decoding temperature controls the entropy of the conditional distribution $F(\cdot \mid W, X)$ from which the LLM surrogates are drawn \citep[see, e.g.,][]{holtzman2020curious}. At zero temperature $F$ is degenerate at the mode, while as temperature increases the distribution approaches a uniform over the token vocabulary.

Temperature thereby affects the per-draw noise variance $\sigma_\varepsilon^2$ in the surrogate, and along with that the reliability ratio $R = \mathrm{Var}(\theta \mid X)/(\mathrm{Var}(\theta \mid X) + \sigma_\varepsilon^2)$ of Proposition~\ref{prop:attenuation}. Two effects are in opposition here: Too low a temperature collapses $F$ toward its mode and can shrink the signal variance $\mathrm{Var}(\theta \mid X)$, while too high a temperature inflates $\sigma_\varepsilon^2$. The optimal temperature therefore maximizes $R$, which is equivalent to minimizing the noise-to-signal ratio $\lambda_w$. With $K > 1$ draws the noise enters as $\sigma_\varepsilon^2/K$, so a larger $K$ permits a higher temperature at the same reliability.

\subsection{Long-term outcomes}

One limitation with LLM-based A/B testing is that many organizations are ultimately interested in long-term outcomes, as evidenced by the growing body of literature on the topic \citep[e.g.,][]{Athey2025, Imbens2025, Yang2024}. The current state of LLM technology makes it conceivable that LLMs can produce short-term outcomes, such as having clicked a button, that are sufficiently aligned with human preferences. However, estimating long-term outcomes---what organizations truly care about---from short-term outcomes in experiments is already challenging with humans alone, and adding LLMs only compounds the difficulty.

Still, classical surrogates that aim to predict long-term outcomes from short-term proxies \citep{athey2019using, kallus2020role, Athey2025} offer an important perspective in making LLM-based A/B testing more aligned with organizational goals. Our framework clarifies that using LLMs to A/B test long-term effects requires a two-step surrogate approach: LLM-based metrics as surrogates for short-term outcomes, and short-term outcomes as surrogates for long-term outcomes.

Consequently, decisions based on LLM-generated outcomes align with long-term objectives only under a compounded set of surrogate assumptions, which are even stricter and more difficult to justify than those required for standard surrogate settings.

\subsection{Pilot allocation}\label{sec:pilot}

A critique of the calibration-plus-diagnostics workflow of Sections~\ref{sec:surrogacy} and~\ref{sec:evaluation} is that a human pilot study partially goes against the intent of LLM-based A/B testing, which is to conserve testing capacity and avoid exposing users to potentially subpar treatments. Taking the argument to the limit, a sufficiently large pilot study is just an ordinary A/B test, and so the LLM has added nothing.

However, this critique misses that the decision to run a pilot, and how large to make it, involves a risk--reward trade-off. Without a pilot, the experimenter inevitably relies on the assumptions of Section~\ref{sec:surrogacy} and the subset of diagnostics from Section~\ref{sec:evaluation} that can be evaluated on historical data alone. They thereby face the risk that the LLM A/B test passes these diagnostics yet produces an ATE that leads to the wrong decision. In that case, the experimenter incurs the loss of deploying the wrong treatment. This loss must be (weakly) greater than the loss incurred by running the pilot study in order for abstaining from a pilot to be optimal. With a small pilot, the experimenter obtains a bounded amount of trusted evidence that the LLM A/B test is calibrated to the new treatment of interest, and the experimenter may choose to increase the pilot size to weigh the value of that information against the cost of human A/B testing.

To make this trade-off concrete, suppose the experimenter must decide whether to deploy a candidate treatment to a population of $N$ users on the basis of an estimate of the human ATE $\tau$, and incurs an expected per-user loss $\ell$ whenever the deployed treatment is inferior by more than a decision-relevance threshold $\kappa > 0$. Let $c$ be the per-user cost of pilot exposure, $n_p$ the pilot size, and $\hat\tau_{\mathrm{comb}}(n_p)$ the precision-weighted combination of $\hat\tau$ and the pilot estimate, which reduces to $\hat\tau$ at $n_p = 0$. The expected total loss is
\begin{equation}\label{eq:pilot_loss}
    L(n_p) = n_p\, c + N\, \ell \cdot \Pr\!\big[\,\big|\hat\tau_{\mathrm{comb}}(n_p) - \tau\big| > \kappa\,\big],
\end{equation}
balancing the cost of the pilot against the expected cost of a wrong deployment, and the experimenter sets $n_p^* \in \arg\min_{n_p \ge 0} L(n_p)$.

This shows that the marginal value of a pilot study, and hence $n_p^*$, increases with three quantities. The first is the cost ratio $\ell/c$: the costlier a wrong deployment relative to a pilot user, the more each pilot observation is worth. The second is the width of the sensitivity bound $2B(\mathrm{TV}_0 + \mathrm{TV}_1)$ of Proposition~\ref{prop:sensitivity}, as a wider bound leaves $\hat\tau$ consistent with a larger neighborhood of $\tau$. Therefore, the pilot study's evidence about $\tau$ contributes more. The third is the intrinsic novelty of the treatment relative to those on which $\hat\mu$ was trained. This follows because the surrogacy falsification test of Section~\ref{sec:evaluation} can only be applied to historical data.

As such, the pilot fraction can be optimized as a design parameter given the experimenter's loss function and the results from the diagnostics. The corner solutions are then straightforward: a pilot fraction of zero suits low-stakes treatments with strong diagnostics, whereas a fraction close to one may be more appropriate in high-stakes settings with novel treatments for which the LLM serves only as an idea generator.

When a human pilot is available for the target treatment, the calibrated estimate can be augmented with a residual correction from the pilot outcomes, yielding a prediction-assisted estimator in the spirit of \citet{angelopoulos2023predictionpowered}. Such an augmentation uses the pilot to correct bias in the calibrated estimator without discarding the precision gains from the LLM sample.

\section{Discussion and Future Directions}\label{sec:discussion}

We introduced a surrogacy-based framework for using LLM-generated outcomes to estimate human causal effects. Adapting sur\-ro\-gate-end\-point theory to AI-gen\-er\-at\-ed participants, it makes explicit the assumptions that need to hold for identifying the causal effect of interest, how to carry out estimation, and what are the roles and con\-se\-quences of novel challenges introduced by using LLMs for this purpose, including their inherent sto\-chas\-tic\-ity, the role of prompting, temperature, pilot studies on real users, as well as limitations to long-term outcomes.

A central takeaway is that A/B testing on LLM responses yields correct results only by assumption, whereas A/B testing on humans is correct by design. Only a subset of the necessary assumptions can be checked on historical data, and even those can only falsify a violation rather than confirm that the assumption holds. It is therefore not possible to verify that LLM responses are valid surrogates for any truly novel treatment, implying that the promise is least justified precisely where it would be most beneficial. Practitioners should therefore treat diagnostics as a prerequisite and view LLM outcomes as complementary to, rather than substitutes for, actual experiments.

Several directions remain open for future research. The framework assumes no interference and a single outcome, whereas many experiments involve interacting users or multiple metrics; relaxing these is a natural next step. Combining our estimator with inference methods for model-generated labels would yield confidence intervals that account for both estimation error in $\hat\mu$ and the stochasticity of $Y^*$. Finally, extending to long-term outcomes, where LLMs predict intermediate outcomes that themselves proxy long-term effects, and multi-agent settings raise new challenges that compound identification requirements and estimation uncertainty.

\begin{acks}
    We thank Kamil Ciosek for helpful comments on the methodology, Kyle Kretschman for feedback on the writing, and attendees of the inaugural Workshop on Experimentation for Decision-Making at Columbia Business School for discussions that helped us refine the positioning and contribution.
\end{acks}

\bibliographystyle{ACM-Reference-Format}
\bibliography{library}

\clearpage
\end{document}